\documentclass[11pt]{article}

\usepackage{amssymb}
\usepackage{amsmath}
\usepackage{amsthm}
\usepackage{algorithm}
\usepackage{algorithmic}
\usepackage{setspace}
\usepackage{epstopdf}
\usepackage{framed}
\usepackage{graphicx}
\usepackage{subcaption}
\usepackage[nohead, margin=0in]{geometry}
\usepackage{fullpage}
\newtheorem{thm}{Theorem}[section]

\newtheorem{lem}[thm]{Lemma}

\newtheorem{defn}[thm]{Definition}
\newtheorem{con}[thm]{Conjecture}

\title{Sparse Approximation is Provably Hard under Coherent Dictionaries}

\author{Ali {\c{C}}ivril}

\begin{document}

\maketitle

\begin{abstract}
It is well known that sparse approximation
problem is \textsf{NP}-hard under general dictionaries. Several algorithms
have been devised and analyzed in the past decade under various
assumptions on the \emph{coherence} $\mu$ of the dictionary
represented by an $M \times N$ matrix from which a subset of $k$
column vectors is selected. All these results assume $\mu=O(k^{-1})$.
This article is an attempt to bridge the big gap between the
negative result of \textsf{NP}-hardness under general dictionaries and the
positive results under this restrictive assumption. In particular,
it suggests that the aforementioned assumption might be asymptotically
the best one can make to arrive at any efficient algorithmic result
under well-known conjectures of complexity theory. In establishing the
results, we make use of a new simple multilayered PCP which is
tailored to give a matrix with small coherence combined with our
reduction.
\end{abstract}
%


\section{Introduction}
Given a dictionary $\Phi$ with normalized columns,
represented by an $M \times N$ matrix ($\Phi \in \mathbb{R}^{M
\times N}$) and a target signal $y \in \mathbb{R}^M$ such that the columns
of $\Phi$ span $\mathbb{R}^M$, \emph{sparse
approximation problem} asks to find an approximate representation
of $y$ using a linear combination of at most $k$ atoms, i.e.
column vectors of $\Phi$. This amounts to finding a coefficient
vector $x \in \mathbb{R}^N$ for which one usually solves
\begin{equation}
\label{sparse} \min_{{\|x\|}_0 = k}{\|y-\Phi x\|}_2
\end{equation}

\noindent We name the problem with this standard objective
function (\ref{sparse}) as \textsc{Sparse}. Stated in linear
algebraic terms, it is essentially about picking a $k$-dimensional
subspace defined by $k$ column vectors of $\Phi$ such that the
orthogonal projection of $y$ onto that subspace is as close as
possible to $y$. The reader should note that the problem can be
defined with full generality using notions from functional
analysis (e.g. Hilbert spaces with elements representing
functions), as is usually conceived in signal processing. Indeed,
defined in Hilbert and Banach spaces, it has been studied as
\emph{highly nonlinear approximation} in functional approximation
theory \cite{Temlyakov1,Temlyakov2}. However, we consider linear
algebraic language as any kind of generalization is irrelevant to
our discussion and the negative results we will present can be
readily extended to the general case.

Although mainly studied in signal processing and approximation
theory, sparse approximation problem is of combinatorial nature in
finite dimensions and the optimal solution can be found by
checking all $\binom{N}{k}$ subspaces. It is natural to ask
whether one can do better and the answer partly lies in the fact
that \textsc{Sparse} is \textsf{NP}-hard even to approximate within any
factor \cite{Davis,Natarajan} under general dictionaries. The
intrinsic difficulty of the problem under this objective function
prevents one from designing algorithms which can even
\emph{approximate} the optimal solution. Hence, efforts have been
towards analyzing algorithms working on a restricted set of
dictionaries for which one requires the column vectors to be an
almost orthogonal set, namely an incoherent dictionary. Formally,
one defines the \emph{coherence} $\mu$ of a dictionary $\Phi$ as
$$
\mu(\Phi) = \max_{i\neq j} | \langle\Phi_i,\Phi_j\rangle |
$$
\noindent where $\Phi_i$ and $\Phi_j$ are the $i^{th}$ and
$j^{th}$ columns of $\Phi$, respectively and
$\langle\cdot,\cdot\rangle$ denotes the usual inner product defined on
$\mathbb{R}^{M}$. Recall that the columns of a dictionary have
unit norm. Hence, the coherence $\mu$ takes values in the $[0,1]$
closed interval.

There are roughly three types of algorithmic results regarding
sparse approximation problem as listed below.
\begin{enumerate}
\item Results relating the quality of the solution ${\|y-\Phi
x\|}_2$ produced by the algorithm to the quality of the optimal
solution ${\|y-\Phi x^*\|}_2$ given that $\mu(\Phi)$ is a slowly
growing or decreasing function of $k$.

\item Results showing the rate of convergence of an algorithm for
elements from a specific set related to the dictionary (e.g.
\cite{Liu-T})

\item Results stating conditions under which an algorithm
optimally recovers a signal either via norms of certain matrices
related to the dictionary (e.g. \cite{Tropp}) or the Restricted
Isometry Property (RIP) (e.g. \cite{Davenport}).
\end{enumerate}

The results of this paper are immediately related to the first
kind, which are expressed via Lebesgue-type inequalities as named
by Donoho et al. \cite{Lebesgue}. Accordingly, we shall
define the following:
\begin{defn}
An algorithm is an $(f(k),g(k))$-approximation algorithm for
\textsc{Sparse} under coherence $h(k)$ if it selects a vector $x$
with at most $g(k)$ nonzero elements from the dictionary $\Phi$
with $\mu(\Phi) \leq h(k)$ such that
$$
\|y-\Phi x \|_2 \leq f(k) \cdot \|y-\Phi x^*\|_2
$$
\noindent where $x^*$ is the optimal solution with at most $k$
nonzero elements.
\end{defn}

Orthogonal Matching Pursuit (OMP) is a well studied greedy
algorithm yielding such approximation guarantees. There is also
a slight variant of this algorithm named Orthogonal Least Squares
(OLS). In the last decade, the following results were found in
a series of papers by different authors:

\begin{thm} \cite{Gilbert}
OMP is an $(8\sqrt{k},k)$-approximation algorithm for
\textsc{Sparse} under coherence $\frac{1}{8\sqrt{2}(k+1)}$.
\end{thm}

\begin{thm} \cite{Tropp}
OMP is a $(\sqrt{1+6k},k)$-approximation algorithm for
\textsc{Sparse} under coherence $\frac{1}{3k}$.
\end{thm}

\begin{thm} \cite{Lebesgue}
OMP is a $(24,\lfloor k \log{k}\rfloor)$-approximation algorithm
for \textsc{Sparse} under coherence $\frac{1}{90k^{3/2}}$.
\end{thm}

\begin{thm} \cite{Zheltov}
OMP is a $(3,2^{\lfloor \frac{1}{\delta} \rfloor}k)$-approximation
algorithm for \textsc{Sparse} under coherence
$\frac{1}{14\left(2^{\lfloor \frac{1}{\delta} \rfloor}
k\right)^{1+\delta}}$, for any fixed $\delta > 0$.
\end{thm}

\begin{thm} \cite{Livshitz}
OLS is a $(3,2k)$-approximation algorithm for \textsc{Sparse}
under coherence $\frac{1}{20k}$.
\end{thm}

In this paper, we are particularly interested in an approximation
of the form $(f(k),k)$ which implies a solution to the standard
sparse approximation problem. We will investigate the possibility of
such an approximation with respect to $h(k)$. First, let us
discuss these algorithmic results qualitatively.
First, notice that all the results assume a coherence of
$O(k^{-1})$. Indeed, it is a curious question whether such a
restrictive assumption is needed for approximating the problem.
There is no trivial answer. Another peculiarity is that all the
results are essentially due to OMP (or a slight variant), which is
a simple and intuitive greedy algorithm reminiscent of the greedy
method for the well known Set Cover problem in combinatorial
optimization. This method is optimal due to a result of Feige
\cite{Feige} with respect to approximating the best solution.
In essence, sparse approximation can also be considered as
a covering problem where we want to cover a target vector using
vectors from a given set. Of course, unlike Set Cover which does
not bear any contextual information on the elements, one also needs
to take the linear algebraic content into account. Hence, thinking
in purely analogical manner, one would expect that OMP is
probably the best algorithm one can hope for under the definition
of approximation we have provided and the assumption of $\mu =
O(k^{-1})$ is most likely necessary. Although our results are
not exactly tight and there is still some room for improvement
(algorithmic and/or complexity theoretic), as we will see, this
intuition is correct to a certain extent.

We would like to note that intuitive reasonings about why one
might need coherence have
already been discussed in the literature. The question whether
one needs coherence is explicitly articulated in \cite{Candes}
where it is pointed out that ``for if two columns are
closely correlated, it will be impossible in general to distinguish
whether the energy in the signal comes from one or the other''.
However strong this intuition is, there is no complexity theoretic
barrier for solving the sparse approximation problem under
reasonably coherent dictionaries. There naturally arises the
question of how much coherence makes the problem intractable.
For instance, does there exist a polynomial time algorithm
for sparse approximation under some very small constant or
inverse logarithmic coherence? This is the conceptual issue
we address in this paper. This work can also be seen as a continuation
of our effort to prove hardness results using the standard
PCP tools for problems of linear algebraic nature. The rationale of
the constructions of the current paper are similar to the ones
in our previous results where we prove hardness results for
subset selection problems in matrices \cite{CSSP-UG,Volume}.


\subsection{Main Results}
We prove the following two theorems, which state the hardness
of \textsc{Sparse} under coherence as a function of $k$
(the number of atoms to be selected):

\begin{thm}
\label{thm_a} For any constant $\epsilon > 0$, and any function $f$, there is
a sparsity parameter $k$ such that there is no polynomial time $(f(k),k)$-approximation
algorithm for \textsc{Sparse} under coherence $\epsilon$ unless \textsf{P} = \textsf{NP}.
\end{thm}

\begin{thm}
\label{thm_b} For any constants $c > 0$ and $\epsilon > 0$, there is
a sparsity parameter $k$ such that there is no polynomial time $(c,k)$-approximation
algorithm for \textsc{Sparse} under coherence
$k^{-1+\epsilon}$ unless Unique Games is in \textsf{P}.
\end{thm}

Both theorems are the result of the same reduction with
different PCPs. Note that the first theorem rules out all
sorts of approximation. This is due to the fact that we
have perfect completeness in the PCP which implies
$\|y-\Phi x\|=0$ in the YES case of our reduction. Our
construction ensures that this value is strictly greater
than $0$ for the NO case implying that there is no
approximation at all under the assumption \textsf{P} $\neq$ \textsf{NP}.
As for the second theorem, Unique Games is
a problem proposed and conjectured to be \textsf{NP}-hard by Khot
\cite{Khot-UGC}, the so called Unique Games Conjecture
(UGC). Even though the assumption of the theorem is known
to be stronger than the famous \textsf{P} $\neq$ \textsf{NP}, current
algorithmic techniques fall short of proving Unique Games
to be in \textsf{P}. In fact, there are many strong inapproximability
results assuming UGC. A famous example is the Max-Cut problem for
which a tight hardness result was proven by Khot et al.
\cite{Khot-Maxcut}. Since the PCP implied by the UGC
does not have perfect completeness, it will be the case
that $\|y-\Phi x\|=\epsilon$ for some small positive
$\epsilon$ in the YES case of our reduction. In the NO
case, we will have some explicit constant value for
$\|y-\Phi x\|$, thereby ruling out a constant factor
approximation only, unlike the strong result in the
first theorem.

We would like to underline that our reductions do not imply
any hardness results for compressed sensing. The dictionaries
we construct are very special deterministic matrices which
do not seem to be of any use in this field. We only make
sure that the column vectors of the dictionary span the
whole space which is the standard assumption in sparse
approximation. In fact, it is an interesting challenge to
come up with complexity theoretic barriers for compressed
sensing, say parameterized with the restricted isometry
constant. Such investigations however, are beyond the scope of this
paper.
\section{Motivation: The Two Layered PCP}
As is usual in most hardness results, our starting point
is the PCP theorem \cite{Pcp,Pcp2} combined with Raz's parallel
repetition theorem \cite{Raz}. In this section, we define the
constraint satisfaction problem implied by these two theorems
and prove a preliminary result showing the hardness of sparse
approximation under coherence $1/2+\epsilon$. The constraint
satisfaction problem that we will make use of should satisfy
an extra property called ``smoothness'', which is usually not
implied by the PCP theorem together with Raz's parallel
repetition theorem. However, we find it convenient to state
the usual construction first, and then point out its
deficiencies and why we need smoothness. This will also allow
the reader to follow the main reasoning of our actual reduction
and why the Unique Games Conjecture is involved in the second
theorem. We then prove our main theorems using multilayered
PCPs in the next two sections.

The constraint satisfaction problem of our interest is often
called the Label Cover problem, first introduced in
\cite{LC}. A Label Cover instance $L$ is defined as
follows: $L = (G(V,W,E),\Sigma_V, \Sigma_W,\Pi)$ where

\begin{itemize}
\item $G(V,W,E)$ is a regular bipartite graph with vertex sets $V$
and $W$, and the edge set $E$. \vspace{1mm}

\item $\Sigma_V$ and $\Sigma_W$ are the label sets associated with
$V$ and $W$, respectively. \vspace{1mm}

\item $\Pi$ is the collection of constraints on the edge set,
where the constraint on an edge $e$ is defined as a function
$\Pi_e\colon \Sigma_V \rightarrow \Sigma_W$.
\end{itemize}

The problem is to satisfy as many constraints
as possible by finding an assignment $A\colon V \rightarrow \Sigma_V$,
$A\colon W \rightarrow \Sigma_W$. A constraint $\Pi_e$ is said
to be satisfied if $\Pi_e(A(v))=A(w)$ where $e=(v,w)$.
Starting from a MAX3-SAT instance
in which each variable occurs exactly $5$ times, one can define
a (fairly well known) reduction to the Label Cover problem.
Applying what is called a parallel repetition to the instance
at hand, one can then get a new Label Cover instance $L$
for which $|V|=(5n/3)^u$, $|W| = n^u$, $|E|=(5n)^u$,
$|\Sigma_V|=7^u$, $|\Sigma_W|=2^u$, the degree of the vertices
in $V$ is $d_V=3^u$, and the degree of the vertices in $W$ is
$d_W=5^u$, where $n$ is the number of variables in the MAX3-SAT
instance and $u$ is the number of parallel repetitions. The
following is a standard result.

\begin{thm} (PCP theorem \cite{Pcp,Pcp2} and Raz's parallel repetition theorem \cite{Raz})
\label{Raz} Let $L$ be given as above. There exists a universal
constant $\gamma > 0$ such that for every large enough constant
$u$, it is \textsf{NP}-hard to distinguish between the following two cases:

\begin{itemize}
\item YES. There is an assignment $A\colon V \rightarrow \Sigma_V$,
$A\colon W \rightarrow \Sigma_W$ such that $\Pi_e$ is satisfied for
all $e \in E$.

\item NO. No assignment can satisfy more than a fraction
$2^{-\gamma u}$ of the constraints in $\Pi$.
\end{itemize}
\end{thm}

The idea of our reduction is simple. The vector $y$ we want
to approximate is composed of all $1$s. In the YES case,
we want to ensure that there are $k$ column vectors in $\Phi$
whose linear combinations can ``cover'' all the coordinates
of $y$. In the NO case, we want to have that there are no
$k$ column vectors that can cover all the coordinates.
Hence, consider the following reduction from Label Cover to
\textsc{Sparse} inspired by a well known reduction proving
hardness for the Set Cover problem \cite{Lund-Yannakakis}:
We first define a set of vectors $x_1, x_2, \hdots, x_{2^u}$
consisting of only $0$s and $1$s, and then normalize them.
The set also satisfies the property that the dot product
of any two vectors in it is $1/2$. It is not difficult to
construct such a vector set; the rows of a Hadamard matrix
of size $2^{u+1}$ readily provides one. Specifically, in all
reductions we present from now on, we consider a Hadamard matrix
of size $2^{x+1}$ where we need at most $2^x$ distinct codewords
covering half of the coordinates. This is by the fact that one
of the rows of the Hadamard matrix is all $1$s and should
be discarded. Thus, the remaining $2^{x+1}-1>2^x$ rows suffice
for our purpose. Given $L$ as above, we define $2^u$
column vectors for each $w \in W$, one for each label in
$\Sigma_W$. Let the column vectors defined for $w$ be $\Phi_{w,1},
\hdots, \Phi_{w,2^u}$. These column vectors have $(5n)^u$
disjoint blocks, one for each edge in $E$. For each
$e \in E$ which is incident to $w \in W$, the block of
$\Phi_{w,i}$ corresponding to $e$ consists of a specified
vector, say $x_i$. The blocks of $\Phi_{w,i}$ corresponding
to the edges which are not incident to $w$ are all zeros.
Finally, we normalize $\Phi_{w,i}$ so that it has norm
$1$ and set $k = |V|+|W| = (5n/3)^u+n^u$.

The reduction performs a similar operation to the left hand
side of the bipartite graph. The usual practice in covering
type problems is to define on $V$, the complement of the
vectors on $W$ in accordance with the projection function
$\Pi$ of $L$, where taking a complement of a vector in
this context means inverting its $0$s and $1$s. We refer the
reader to Figure~\ref{fig:graph} and Figure~\ref{fig:matrix}
for an illustration of the column vectors defined corresponding
to a simple  Label-Cover instance where the complement of a
vector $x$ is denoted by $\overline{x}$. Here, there are
$|E| = (5n)^u$ blocks in each column vector where each block
is given by a specific Hadamard code.  Thus, the number of
coordinates of a single column vector is $M = 2^{u+1}(5n)^u$.
This construction ensures that in the YES case of $L$ where
there is a labeling
satisfying all the edges, one can prove the existence of
a subset of column vectors complementing each other on
all the coordinates, thereby ``covering'' the target vector $y$.
However, there arises the following typical problem due
to the fact that the label sizes $|\Sigma_V| = 7^u$ and
$|\Sigma_W|=2^u$ are not equal: In order to make sure that
the main idea of the reduction works, one is also inherently
forced to define a total of $2^u$ distinct vectors for a
given $v \in V$, just as we did for the vertices in $W$.
But, there are $7^u$ labels in $\Sigma_V$. This may result
in more than one copies of the same column vector meaning
that the coherence of the dictionary may be very large,
in particular $1$, since there is no
restriction on how the projection function $\Pi$ distributes
$7^u$ labels in $\Sigma_V$ to $2^u$ labels in $\Sigma_W$.
(Note that in this case, we have that the total number of
columns $N = 7^u|V| + 2^u|W| = (35n/3)^u+(2n)^u$).
This is the key point where our reduction deviates from
the usual Set Cover type reductions. In order to ensure
a small coherence, we need a more intricate structure or
a stronger assumption. Accordingly, Theorem \ref{thm_a}
is proved using what is called a Smooth Label Cover instance
and Theorem \ref{thm_b} is proved using the Unique Label Cover
instance which appears in the definition of the Unique Games Conjecture.
To be complete, there is one technicality that must be handled:
The column vectors of $\Phi$ must span $\mathbb{R}^M$. This
can easily be satisfied by adding the identity matrix
of size $M \times M$ to the set of all column vectors.

We now give a theorem stating the existence of a two layered
Smooth Label Cover instance from \cite{SmoothLC}.

\begin{thm} [\cite{SmoothLC}]
\label{smooth}
Given a Label Cover instance
$L = (G(V,W,E),\Sigma_V, \Sigma_W,\Pi)$ with
$|\Sigma_V| = 7^{(T+1)u}$ and $|\Sigma_W|= 2^u 7^{Tu}$,
there is an absolute constant $\gamma > 0$ such that for
all integer parameters $u$ and $T$, it is \textsf{NP}-hard to distinguish
between the following two cases:

\begin{itemize}
\item YES. There is an assignment $A\colon V \rightarrow
\Sigma_V$, $A\colon W \rightarrow \Sigma_W$ such that $\Pi_e$ is
satisfied for all $e \in E$.

\item NO. No assignment can satisfy more than a fraction
$2^{-\gamma u}$ of the constraints in $\Pi$.
\end{itemize}

\noindent Furthermore, $L$ has the following ``smoothness
property''. For every $v \in V$, and $a_1,a_2 \in \Sigma_V$
such that $a_1 \neq a_2$, if $w$ is a randomly chosen
neighbor of $v$, then

$$
Pr_w \left[ \Pi_{(v,w)}(a_1) = \Pi_{(v,w)}(a_2) \right] \leq \frac{1}{T}.
$$

\noindent One can further assume that $|V| = O(n^{Tu})$, $|W| = O(n^{Tu})$
and $G(V,W,E)$ is a regular bipartite graph where the
degree of the vertices in $V$ is $d_V = {(T+1)u \choose u} 3^u$
and the degree of the vertices in $W$ is $d_W = 5^u$.
\end{thm}

Given the Smooth Label Cover instance whose existence is
guaranteed by the theorem above, our reduction is similar
to the one we described for the Label Cover instance. The
vector $y$ we want to approximate is an all $1$s vector.
In order to define $\Phi$, let $x_1, x_2, \hdots, x_{2^u 7^{Tu}}$
be a set of vectors consisting of only  $0$s and $1$s.
Perform a normalization on these vectors, i.e.
$x_i \leftarrow x_i/\|x_i\|_2$. Note that a coordinate of
$x_i$ is either $0$ or a constant depending on $u$ and $T$,
say $\eta$. We also have that the dot product of any two
vectors in this set is $1/2$. As noted earlier, the rows
of a Hadamard matrix of size $2^{u+1} 7^{Tu}$ (which is
a constant) satisfy this property. Given $L$ as above, we
define $2^u 7^{Tu}$ column vectors for each $w \in W$, one
for each label in $\Sigma_W$. Let the column vectors
defined for $w \in W$ be $\Phi_{w,1}, \hdots, \Phi_{w,2^u 7^{Tu}}$.
These column vectors have $|E| = O(n^{Tu})$ disjoint blocks,
one for each edge in $E$. For each $e \in E$ which is
incident to $w \in W$, the block of $\Phi_{w,i}$
corresponding to $e$ consists of a specified vector $x_i$.
The blocks of $\Phi_{w,i}$ corresponding to
the edges which are not incident to $w$ are all zeros.
Finally, we normalize $\Phi_{w,i}$ so that it has norm $1$.

Similar to $W$, we define a column vector for each
vertex-label pair on $V$. Specifically, let $\overline{x_i}$
be the vector such that its $j$th coordinate
$\overline{x_{ij}} = 0$ if the $j$th coordinate $x_{ij}$
of $x_i$ is $\eta$, and $\overline{x_{ij}} = \eta$ if
$x_{ij} = 0$. Let the column vectors defined for
$v \in V$ be $\Phi_{v,1}, \hdots, \Phi_{v,7^{(T+1)u}}$.
These column vectors have $|E| = O(n^{Tu})$ disjoint blocks,
one for each edge in $E$. For each $e \in E$ which is
incident to $v \in V$, the block of $\Phi_{v,i}$ corresponding
to $e$ is the vector $\overline{x_{\Pi_e(i)}}$. The blocks
of $\Phi_{v,i}$ corresponding to the edges which are not
incident to $v$ are all zeros. We normalize
$\Phi_{v,i}$ so that it has norm $1$ and set $k=|V|+|W| = O(n^{Tu})$.
We finally add an identity matrix of size $M \times M$ to $\Phi$
in order to make sure that its column vectors span $\mathbb{R}^M$.
As for the size of the reduction where we have
$\Phi \in \mathbb{R}^{M \times N}$ and $y \in \mathbb{R}^M$,
it is clear that $M = 2^{u+1}7^{Tu}|E| = 2^{u+1}7^{Tu}5^u|W|$ and
$N = 7^{(T+1)u}|V|+2^u7^{Tu}|W|$, which are
both of polynomial size for $T$ and $u$ constant since
$|V|$ and $|W|$ are both $O(n^{Tu})$. The mechanics of the
reduction is best explained pictorially. For this reason,
In Figure~\ref{fig:graph}, we show a small part of a Smooth
Label Cover instance. The corresponding matrix $\Phi^T$ is
given in Figure~\ref{fig:matrix}.

We now prove the completeness
of the reduction. Suppose that there is an assignment of the
vertices in $V$ and $W$ which satisfies all the edges in $E$.
Given such an assignment, take the corresponding column vectors
of $\Phi$, a total of $k=|V|+|W|$ vectors. By the reduction,
given a block corresponding to a specific edge $e = (v,w)$, we
have that the column vectors selected from $v$ and $w$ cover
all the coordinates allocated to $e$, i.e. the nonzero values
in $v$ and $w$ complement each other. Thus, the selected $k$
column vectors cover all the $M$ coordinates and there
exists a vector $x \in \mathbb{R}^N$ with $k$ nonzero entries
satisfying $\Phi x = y$.

Suppose now that there is no assignment of the vertices that
satisfies all the edges. First, note that in order to cover
all the $M$ coordinates, one needs to select exactly $1$
column vector from each vertex in $V$ and $W$. Because, if
there is a vertex $v$ from which no column vector is selected,
one cannot cover the blocks of edges incident to $v$ by
selecting one vector from the neighbors of $v$ (The rows of
a Hadamard matrix cover only half of the coordinates). Hence,
one needs to select multiple vectors to cover these blocks
and it follows that all the coordinates cannot be covered by
selecting $k$ column vectors. But, even if one selects $1$
vector from each vertex, the unsatisfied edges cannot be
covered by the definition of the reduction. Hence, for all
$x \in \mathbb{R}^N$, we have $\|y - \Phi x\|_2 > 0$. This
shows that \textsc{Sparse} cannot be approximated at all given
the conditions of $\Phi$ provided by our reduction.

It remains to see what the coherence of $\Phi$ is. First,
note that the dot product of any two column vectors belonging
to different vertices in $V$ or $W$ is $0$ since they do not
share any blocks. Besides, the dot product between a column
vector in $V$ and another column vector in $W$ is clearly
(much) less than $1/2$ since they share at most one block.
There remain two basic cases to check. The column
vectors defined for a vertex $w \in W$ have half of their
coordinates overlap with each other and hence all the pairs have
dot product $1/2$. The situation is a little more involved
for a vertex $v \in V$. Given $\Phi_{v,i}$ and $\Phi_{v,j}$,
by the smoothness property of the Label Cover instance, we have
that $\overline{x_{\Pi_e(i)}} = \overline{x_{\Pi_e(j)}}$ for
at most $1/T$ fraction of the edges that are incident to $v$.
Thus, the extra contribution we have for the dot product of
two column vectors defined for $v$, because of having the
same entries on some of the blocks, is at most $1/T$ (Note that,
this is not satisfied by the usual Label Cover instance).
It follows that the maximum dot product of any two column vectors
of $\Phi$ is at most $1/2+1/T$. Hence,
\textsc{Sparse} is \textsf{NP}-hard under dictionaries with coherence
$1/2+\epsilon$ for arbitrarily small constant $\epsilon$ upon
selecting $T$ arbitrarily large.

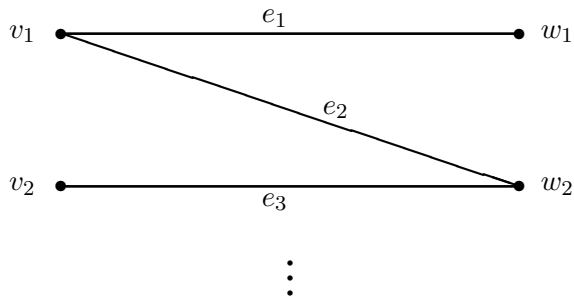
\begin{figure}[h]
\centering
\setlength{\unitlength}{.4in}
\begin{picture}(10,5)(0,0)
\linethickness{1pt} \put(2,4){\circle*{0.15}}
\put(2,4){\line(1,0){6}} \put(2,2){\circle*{0.15}}
\put(2,2){\line(1,0){6}} \thicklines \put(2,4){\line(3,-1){6}}
\put(8,4){\circle*{0.15}} \put(8,2){\circle*{0.15}}
\put(1.5,4){\makebox(0,0){$v_1$}}
\put(1.5,2){\makebox(0,0){$v_2$}}
\put(8.5,4){\makebox(0,0){$w_1$}}
\put(8.5,2){\makebox(0,0){$w_2$}}
\put(4.8,4.2){\makebox(0,0){$e_1$}}
\put(4.8,1.8){\makebox(0,0){$e_3$}}
\put(5.6,3){\makebox(0,0){$e_2$}}
\put(5,1){\circle*{0.08}}
\put(5,0.8){\circle*{0.08}} \put(5,0.6){\circle*{0.08}}
\end{picture}
\caption{A part of a simple bipartite graph representing a
Smooth Label Cover instance} \label{fig:graph}
\end{figure}

\begin{figure}[h]
\centering
\setlength{\unitlength}{.25in}
\begin{picture}(15,15)(0,0)
\linethickness{0.5pt} \put(1.5,13){\line(1,0){12.5}}
\put(1.5,12){\line(1,0){12.5}} \put(1.5,11){\line(1,0){12.5}}
\put(1.5,9){\line(1,0){12.5}} \put(1.5,8){\line(1,0){12.5}}
\put(1.5,6){\line(1,0){12.5}} \put(1.5,5){\line(1,0){12.5}}
\put(1.5,3){\line(1,0){12.5}} \put(1.5,2){\line(1,0){12.5}}

\put(1.5,13){\line(0,-1){2}} \put(1.5,9){\line(0,-1){1}}
\put(1.5,6){\line(0,-1){1}} \put(1.5,3){\line(0,-1){1}}

\put(4,13.5){\line(0,-1){12}} \put(6.5,13.5){\line(0,-1){12}}
\put(9,13.5){\line(0,-1){12}}

\put(14,13){\line(0,-1){2}} \put(14,9){\line(0,-1){1}}
\put(14,6){\line(0,-1){1}} \put(14,3){\line(0,-1){1}}

\put(0.8,12.5){\makebox(0,0){$\Phi_{v_1,1}$}}
\put(0.8,11.5){\makebox(0,0){$\Phi_{v_1,2}$}}
\put(0.8,8.5){\makebox(0,0){$\Phi_{v_2,1}$}}
\put(0.8,5.5){\makebox(0,0){$\Phi_{w_1,1}$}}
\put(0.8,2.5){\makebox(0,0){$\Phi_{w_2,1}$}}

\put(2.7,13.5){\makebox(0,0){$e_1$}}
\put(5.2,13.5){\makebox(0,0){$e_2$}}
\put(7.7,13.5){\makebox(0,0){$e_3$}}

\put(2.7,12.5){\makebox(0,0){$\overline{x_{\Pi_{e_1}(1)}}$}}
\put(5.2,12.5){\makebox(0,0){$\overline{x_{\Pi_{e_2}(1)}}$}}
\put(7.7,12.5){\makebox(0,0){$\overrightarrow{0}$}}
\put(11.5,12.5){\makebox(0,0){$\overrightarrow{0}$}}

\put(2.7,11.5){\makebox(0,0){$\overline{x_{\Pi_{e_1}(2)}}$}}
\put(5.2,11.5){\makebox(0,0){$\overline{x_{\Pi_{e_2}(2)}}$}}
\put(7.7,11.5){\makebox(0,0){$\overrightarrow{0}$}}
\put(11.5,11.5){\makebox(0,0){$\overrightarrow{0}$}}

\put(2.7,8.5){\makebox(0,0){$\overrightarrow{0}$}}
\put(5.2,8.5){\makebox(0,0){$\overrightarrow{0}$}}
\put(7.7,8.5){\makebox(0,0){$\overline{x_{\Pi{{e_3}(1)}}}$}}
\put(11.5,8.5){\makebox(0,0){$\overrightarrow{0}$}}

\put(2.7,5.5){\makebox(0,0){$x_1$}}
\put(5.2,5.5){\makebox(0,0){$\overrightarrow{0}$}}
\put(7.7,5.5){\makebox(0,0){$\overrightarrow{0}$}}
\put(11.5,5.5){\makebox(0,0){$\overrightarrow{0}$}}

\put(2.7,2.5){\makebox(0,0){$\overrightarrow{0}$}}
\put(5.2,2.5){\makebox(0,0){$x_1$}}
\put(7.7,2.5){\makebox(0,0){$x_1$}}
\put(11.5,2.5){\makebox(0,0){$\overrightarrow{0}$}}

\put(2,9.8){\circle*{0.07}} \put(2,10){\circle*{0.07}}
\put(2,10.2){\circle*{0.07}}

\put(2,6.8){\circle*{0.07}} \put(2,7){\circle*{0.07}}
\put(2,7.2){\circle*{0.07}}

\put(2,3.8){\circle*{0.07}} \put(2,4){\circle*{0.07}}
\put(2,4.2){\circle*{0.07}}

\put(11,13.3){\circle*{0.07}} \put(11.3,13.3){\circle*{0.07}}
\put(11.6,13.3){\circle*{0.07}}

\put(11,9.3){\circle*{0.07}} \put(11.3,9.3){\circle*{0.07}}
\put(11.6,9.3){\circle*{0.07}}

\put(11,6.3){\circle*{0.07}} \put(11.3,6.3){\circle*{0.07}}
\put(11.6,6.3){\circle*{0.07}}

\put(11,3.3){\circle*{0.07}} \put(11.3,3.3){\circle*{0.07}}
\put(11.6,3.3){\circle*{0.07}}
\end{picture}

\caption{The resulting (row) vectors in \textsc{Sparse} instance computed
from the graph in Figure \ref{fig:graph} by our reduction}
\label{fig:matrix}
\end{figure}
\section{Reduction from the Multilayered Smooth Label Cover}
\begin{figure}[h]
\begin{subfigure}{.4\textwidth} 
  \centering
  \includegraphics[width=.296\linewidth]{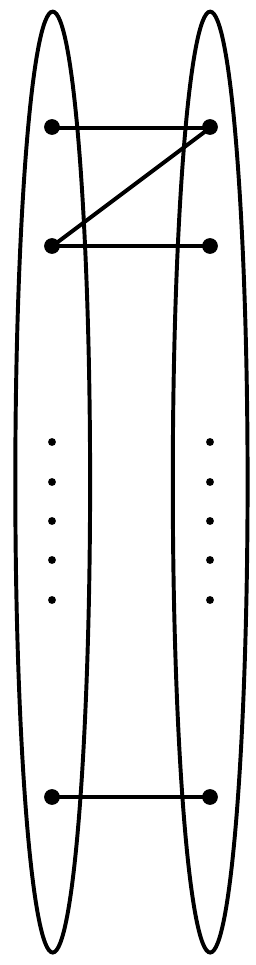} 
  \caption{The two layered PCP}
  \label{case1-1}
\end{subfigure}
\begin{subfigure}{.5\textwidth} 
  \centering
  \includegraphics[width=\linewidth]{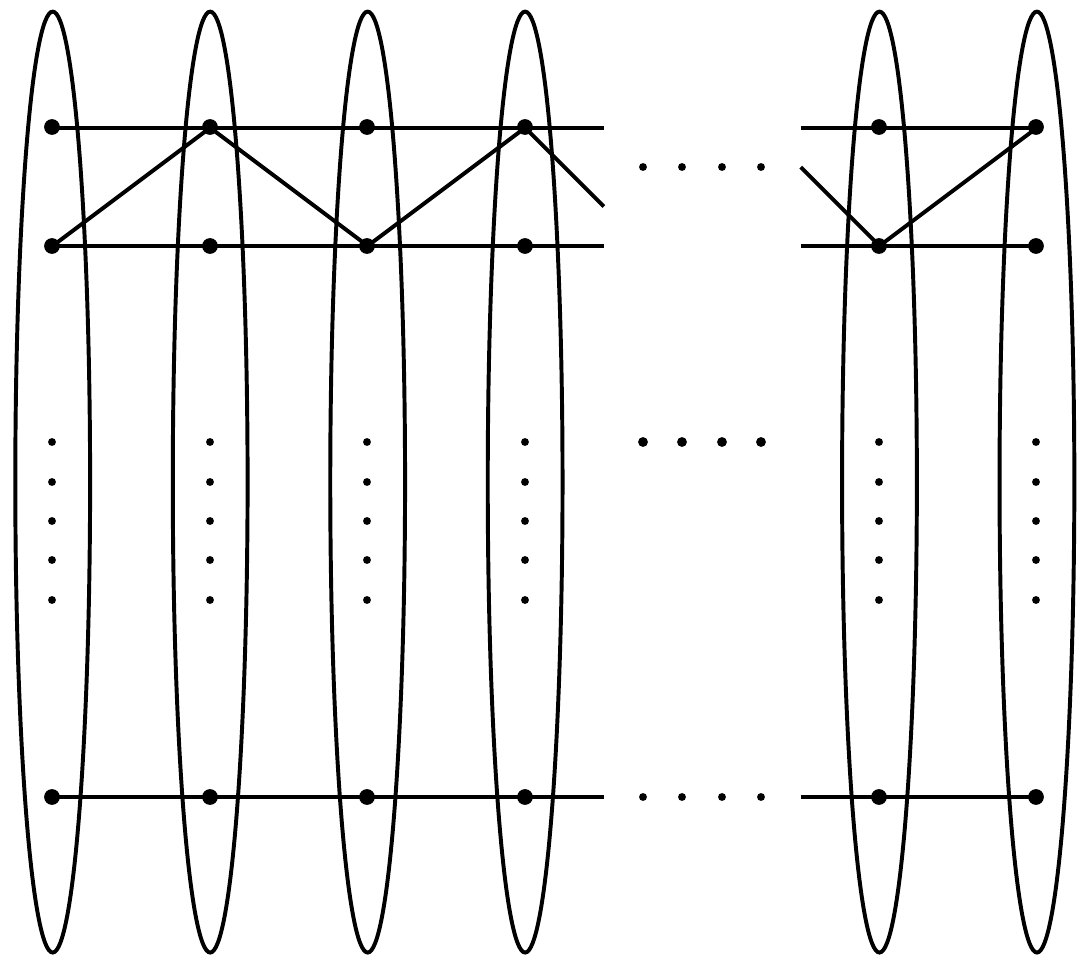}
  \caption{The corresponding multilayered PCP that we make use of}
  \label{case1-2}
\end{subfigure}
\caption{}
\label{fig:multilayered}
\end{figure}

In this section, we define a simple multilayered PCP starting
from a two layered Smooth Label Cover instance. Let
$L = (G(V,W,E),\Sigma_V, \Sigma_W,\Pi)$ be such a
Label Cover instance defined as in the previous section.
An $\ell$-layered Smooth Label Cover instance $L^{\ell}$ is
defined as follows: $L^{\ell} = (G^{\ell}(X_1, \hdots, X_{\ell},E^{\ell}),
\Sigma_V, \Sigma_W, \Pi^{\ell})$ where $\ell$ is an even
positive integer and

\begin{itemize}
\item $G^{\ell}(X_1, \hdots, X_{\ell},E^{\ell})$ is a multipartite
hyper-graph with vertex sets $X_1, X_2, \hdots, X_{\ell}$ and the
hyper-edge set $E$. \vspace{1mm}

\item $X_i = V$ for odd $i$, and $X_i = W$ for even $i$, where
$1 \leq i \leq \ell$. \vspace{1mm}

\item $E^{\ell}$ consists of all hyper-edges of the form
$(v,w,v,w,\hdots,v,w)$ on $\ell$ vertices belonging to $X_1, \hdots, X_{\ell}$,
respectively where $v \in V$, $w \in W$ and $(v,w) \in E$.

\item $\Pi^{\ell}$ is the collection of constraints on the
hyper-edge set, where the constraint on a hyper-edge
$e=(x_1,x_2,\hdots,x_{\ell-1},x_{\ell})=(v,w,\hdots,v,w)$ is itself
defined as a collection of $\ell-1$ constraints of the form
$\Pi_{(x_{2i-1},x_{2i})}\colon \Sigma_V \rightarrow \Sigma_W$ for $i=1,\hdots,\ell/2$, and
$\Pi_{(x_{2i+1},x_{2i})}\colon \Sigma_V \rightarrow \Sigma_W$ for $i=1,\hdots,\ell/2-1$, where
$x_i \in X_i$ for $i=1,\hdots,\ell$, and all the constraints are
identical to $\Pi_{(v,w)}$.
\end{itemize}

As usual, one is asked to find an assignment for all the
vertices of the instance, i.e. construct functions
$A_{2i-1}\colon X_i \rightarrow \Sigma_V$, and
$A_{2i}\colon X_i \rightarrow \Sigma_W$ for $i=1, \hdots, \ell/2$,
where the goal is to satisfy as many hyper-edges as possible.
In this case however, there are two definitions of
satisfiability. A constraint $\Pi_e$ is said to be
\emph{strongly satisfied} if all of the $\ell-1$ constraints
defined for $e$ are satisfied. Otherwise, it is said to be
\emph{weakly satisfied} if only a fraction of its constraints
are satisfied. This is the standard terminology introduced in
Feige's paper about the Set Cover problem \cite{Feige} and
the difference between strong and weak satisfiability is
important in that context. However, we will be only arguing
about the strong satisfiability in this work as it is sufficient
to derive our hardness results.

The foregoing definition already defines a fairly
straightforward reduction from the Smooth Label Cover instance
to the multilayered instance: Just replace each layer
alternatively by $V$ and $W$, define a hyper-edge in
$e \in E^{\ell}$ for each edge $e'=(v,w) \in E$ such that it
contains all the copies of $v$ and $w$ in alternating layers
(i.e. $e'$ uniquely defines $e$), and finally define all the
$\ell-1$ projection functions of a hyper-edge $e$ to be the
projection function $\Pi_{e'}$ on the corresponding edge $e'$
of the original instance. The reduction is shown in
Figure~\ref{fig:multilayered} where each hyper-edge of the
multilayered instance is defined on exactly $\ell$ vertices,
three of them straight, one zigzagging. Given this reduction,
it is not difficult to see that if there is an assignment for
$L$ which satisfies all the constraints, then the same assignment
\emph{strongly} satisfies all the constraints of $L^{\ell}$.
Similarly, if no assignment satisfies more than a fraction
$2^{-\gamma u}$ of the constraints in $L$, then no assignment
\emph{strongly} satisfies more than a fraction $2^{-\gamma u}$
of the constraints in $L^{\ell}$.

The rows of a Hadamard matrix was sufficient to give a set
of vectors with pairwise dot products $1/2$. In order to get
a coherence of arbitrarily small constant value starting
from the $\ell$-layered instance defined above, we need a
more general construction. The construction which we describe
below is essentially inspired by that of partition systems
mentioned by Feige \cite{Feige}. However, it is not clear
whether these systems satisfy conditions on the coherence
that we require. They are tailored to prove hardness for Set
Cover. Hence, we describe our construction from the first
principles. This will also allow the reader to see why the
construction works at an intuitive level. To this aim, we
define an \emph{incoherent vector system} $V(\ell,d)$. It
has the following properties:

\begin{enumerate}
\item It is a set $V(\ell,d)$ of $\ell d$ normalized vectors of
dimension $\ell^d$ with exactly $\ell^{d-1}$ nonzero coordinates of the
same positive value.

\item $V(\ell,d) = \bigcup_{i=1}^d V(i,\ell,d)$ where
$V(i,\ell,d)$ is a set of $\ell$ vectors for $i = 1, \hdots, d$,
and $V(i,\ell,d) \cap V(j,\ell,d) = \emptyset$ for $i \neq j$.

\item Any two distinct vectors in $V(i,\ell,d)$ have dot product $0$ for $i
= 1, \hdots, d$. Furthermore, the indices of the nonzero
coordinates of the vectors in $V(i,\ell,d)$ cover the set $\{1,
\hdots, \ell^d\}$.

\item For a vector $x \in V(i,\ell,d)$ and $y \in V(j,\ell,d)$,
the dot product of $x$ and $y$ is $1/\ell$, where $i \neq j$.
\end{enumerate}

\noindent For explicitly denoting vectors in the system, we let
$$
V(i,\ell,d) = \{V_1(i,\ell,d), V_2(i,\ell,d), \hdots, V_{\ell}(i,\ell,d)\},
$$

\noindent where the vectors are ordered lexicographically with respect to
the nonzero entries in their coordinates. The result of the
reduction of this section is provided by the following lemma.
\begin{lem}
\label{existence_1} There exists an explicit deterministic
construction of an incoherent vector system $V(\ell, d)$
for all $1 \leq d \leq \ell$.
\end{lem}

\begin{proof}
We will construct the desired vectors by describing a recursive
procedure. Consider first the $\ell$ strings of length $\ell$ each
with a distinct coordinate being $1$ and all other coordinates $0$
(e.g. for $\ell = 3$, the strings are $100, 010$ and $001$). Let
these ``seed'' strings be $u_1, u_2, \hdots u_{\ell}$ in
lexicographic order. For each of these strings, concatenate
$\ell^{d-1}$ of them side by side to form $\ell$ strings of
length exactly $\ell^d$. To give a better idea of the
procedure, we will get the following strings for $\ell = d = 3$:
\begin{eqnarray*}
v_1 = 100100100100100100100100100 \\
v_2 = 010010010010010010010010010 \\
v_3 = 001001001001001001001001001
\end{eqnarray*}

\noindent Normalizing these $\ell$ strings, i.e. multiplying
each coordinate by $\ell^{(1-d)/2}$ we obtain the vectors

$$V_1(d, \ell, d), V_2(d, \ell, d), \hdots , V_{\ell}(d, \ell, d),$$

\noindent which form the set $V(d, \ell, d)$. This makes a total of
$\ell$ vectors of the incoherent vector system at the lowest
level of our recursion. Note also that these vectors have
pairwise dot product $0$.

In order to construct the vectors of $V(d-1, \ell, d)$, we
create a set of new strings by ``expanding'' each coordinate
of the seed string $u_i$, namely by repeating their coordinates
exactly $\ell$ times in place. Thus, we get a new set of
$\ell$ strings of length $\ell^2$ (e.g. for $\ell = 3$, the
aforementioned strings are $111000000, 000111000$ and
$000000111$). For each of these strings, concatenating
$\ell^{d-2}$ of them side by side, we get $\ell$ strings of
length $\ell^d$. Normalizing these strings, we obtain the
vectors of the set $V(d-1, \ell, d)$. Note that the pairwise dot
products of the vectors in this set are also $0$. Furthermore,
the dot product of a vector in $V(d-1, \ell, d)$ and
another vector in $V(d, \ell, d)$ is exactly $1/\ell$
as there are $\ell^{d-2}$ common nonzero coordinates between
two such vectors and the nonzero entries are all $\ell^{(1-d)/2}$.

\begin{figure}[t]
\centering
\includegraphics[width=\linewidth]{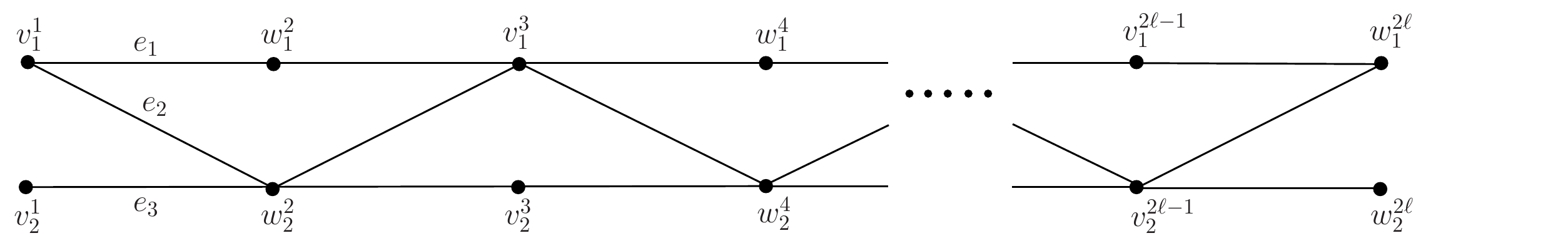}
\caption{A part of a simple bipartite graph representing a
multilayered Smooth Label Cover instance} \label{fig:multi}
\end{figure}

\begin{figure}
\centering
\setlength{\unitlength}{.25in}
\begin{picture}(22,26)(0,0)
\linethickness{0.5pt} \put(1.5,24){\line(1,0){19.5}}
\put(1.5,23){\line(1,0){19.5}} \put(1.5,22){\line(1,0){19.5}}
\put(1.5,20){\line(1,0){19.5}} \put(1.5,19){\line(1,0){19.5}}
\put(1.5,17){\line(1,0){19.5}} \put(1.5,16){\line(1,0){19.5}}
\put(1.5,14){\line(1,0){19.5}} \put(1.5,13){\line(1,0){19.5}}
\put(1.5,11){\line(1,0){19.5}} \put(1.5,10){\line(1,0){19.5}}
\put(1.5,8){\line(1,0){19.5}} \put(1.5,7){\line(1,0){19.5}}
\put(1.5,5){\line(1,0){19.5}} \put(1.5,4){\line(1,0){19.5}}
\put(1.5,2){\line(1,0){19.5}} \put(1.5,1){\line(1,0){19.5}}

\put(1.5,24){\line(0,-1){2}} \put(1.5,20){\line(0,-1){1}}
\put(1.5,17){\line(0,-1){1}} \put(1.5,14){\line(0,-1){1}}

\put(21,24){\line(0,-1){2}} \put(21,20){\line(0,-1){1}}
\put(21,17){\line(0,-1){1}} \put(21,14){\line(0,-1){1}}

\put(1.5,11){\line(0,-1){1}} \put(1.5,8){\line(0,-1){1}}
\put(1.5,5){\line(0,-1){1}} \put(1.5,2){\line(0,-1){1}}

\put(5.5,24.5){\line(0,-1){24}} \put(9.5,24.5){\line(0,-1){24}}
\put(13.5,24.5){\line(0,-1){24}}

\put(21,11){\line(0,-1){1}} \put(21,8){\line(0,-1){1}}
\put(21,5){\line(0,-1){1}} \put(21,2){\line(0,-1){1}}

\put(0.8,23.5){\makebox(0,0){$\Phi_{v_1^1,1}$}}
\put(0.8,22.5){\makebox(0,0){$\Phi_{v_1^1,2}$}}
\put(0.8,19.5){\makebox(0,0){$\Phi_{v_2^1,1}$}}
\put(0.8,16.5){\makebox(0,0){$\Phi_{w_1^2,1}$}}
\put(0.8,13.5){\makebox(0,0){$\Phi_{w_2^2,1}$}}
\put(0.8,10.5){\makebox(0,0){$\Phi_{v_1^3,1}$}}
\put(0.8,7.5){\makebox(0,0){$\Phi_{v_2^3,1}$}}
\put(0.8,4.5){\makebox(0,0){$\Phi_{w_1^4,1}$}}
\put(0.8,1.5){\makebox(0,0){$\Phi_{w_2^4,1}$}}

\put(3.7,24.5){\makebox(0,0){$e_1$}}
\put(7.7,24.5){\makebox(0,0){$e_2$}}
\put(11.7,24.5){\makebox(0,0){$e_3$}}

\put(3.5,23.5){\makebox(0,0){$V_1(\Pi_{e_1}(1),\ell,d)$}}
\put(7.5,23.5){\makebox(0,0){$V_1(\Pi_{e_2}(1),\ell,d)$}}
\put(11.5,23.5){\makebox(0,0){$\overrightarrow{0}$}}
\put(17,23.5){\makebox(0,0){$\overrightarrow{0}$}}

\put(3.5,22.5){\makebox(0,0){$V_1(\Pi_{e_1}(2),\ell,d)$}}
\put(7.5,22.5){\makebox(0,0){$V_1(\Pi_{e_1}(2),\ell,d)$}}
\put(11.5,22.5){\makebox(0,0){$\overrightarrow{0}$}}
\put(17,22.5){\makebox(0,0){$\overrightarrow{0}$}}

\put(3.5,19.5){\makebox(0,0){$\overrightarrow{0}$}}
\put(7.5,19.5){\makebox(0,0){$\overrightarrow{0}$}}
\put(11.5,19.5){\makebox(0,0){$V_1(\Pi_{e_3}(1),\ell,d)$}}
\put(17,19.5){\makebox(0,0){$\overrightarrow{0}$}}

\put(3.5,16.5){\makebox(0,0){$V_2(1,\ell,d)$}}
\put(7.5,16.5){\makebox(0,0){$\overrightarrow{0}$}}
\put(11.5,16.5){\makebox(0,0){$\overrightarrow{0}$}}
\put(17,16.5){\makebox(0,0){$\overrightarrow{0}$}}

\put(3.5,13.5){\makebox(0,0){$\overrightarrow{0}$}}
\put(7.5,13.5){\makebox(0,0){$V_2(1,\ell,d)$}}
\put(11.5,13.5){\makebox(0,0){$V_2(1,\ell,d)$}}
\put(17,13.5){\makebox(0,0){$\overrightarrow{0}$}}

\put(3.5,10.5){\makebox(0,0){$V_3(\Pi_{e_1}(1),\ell,d)$}}
\put(7.5,10.5){\makebox(0,0){$V_3(\Pi_{e_1}(1),\ell,d)$}}
\put(11.5,10.5){\makebox(0,0){$\overrightarrow{0}$}}
\put(17,10.5){\makebox(0,0){$\overrightarrow{0}$}}

\put(3.5,7.5){\makebox(0,0){$\overrightarrow{0}$}}
\put(7.5,7.5){\makebox(0,0){$\overrightarrow{0}$}}
\put(11.5,7.5){\makebox(0,0){$V_3(\Pi_{e_3}(1),\ell,d)$}}
\put(17,7.5){\makebox(0,0){$\overrightarrow{0}$}}

\put(3.5,4.5){\makebox(0,0){$V_4(1,\ell,d)$}}
\put(7.5,4.5){\makebox(0,0){$\overrightarrow{0}$}}
\put(11.5,4.5){\makebox(0,0){$\overrightarrow{0}$}}
\put(17,4.5){\makebox(0,0){$\overrightarrow{0}$}}

\put(3.5,1.5){\makebox(0,0){$\overrightarrow{0}$}}
\put(7.5,1.5){\makebox(0,0){$V_4(1,\ell,d)$}}
\put(11.5,1.5){\makebox(0,0){$V_4(1,\ell,d)$}}
\put(17,1.5){\makebox(0,0){$\overrightarrow{0}$}}

\put(2,20.8){\circle*{0.07}} \put(2,21){\circle*{0.07}}
\put(2,21.2){\circle*{0.07}}

\put(2,17.8){\circle*{0.07}} \put(2,18){\circle*{0.07}}
\put(2,18.2){\circle*{0.07}}

\put(2,14.8){\circle*{0.07}} \put(2,15){\circle*{0.07}}
\put(2,15.2){\circle*{0.07}}

\put(2,11.8){\circle*{0.07}} \put(2,12){\circle*{0.07}}
\put(2,12.2){\circle*{0.07}}

\put(2,8.8){\circle*{0.07}} \put(2,9){\circle*{0.07}}
\put(2,9.2){\circle*{0.07}}

\put(2,5.8){\circle*{0.07}} \put(2,6){\circle*{0.07}}
\put(2,6.2){\circle*{0.07}}

\put(2,2.8){\circle*{0.07}} \put(2,3){\circle*{0.07}}
\put(2,3.2){\circle*{0.07}}

\put(16,24.3){\circle*{0.07}} \put(16.3,24.3){\circle*{0.07}}
\put(16.6,24.3){\circle*{0.07}}

\put(16,20.3){\circle*{0.07}} \put(16.3,20.3){\circle*{0.07}}
\put(16.6,20.3){\circle*{0.07}}

\put(16,17.3){\circle*{0.07}} \put(16.3,17.3){\circle*{0.07}}
\put(16.6,17.3){\circle*{0.07}}

\put(16,14.3){\circle*{0.07}} \put(16.3,14.3){\circle*{0.07}}
\put(16.6,14.3){\circle*{0.07}}

\put(16,11.3){\circle*{0.07}} \put(16.3,11.3){\circle*{0.07}}
\put(16.6,11.3){\circle*{0.07}}

\put(16,8.3){\circle*{0.07}} \put(16.3,8.3){\circle*{0.07}}
\put(16.6,8.3){\circle*{0.07}}

\put(16,5.3){\circle*{0.07}} \put(16.3,5.3){\circle*{0.07}}
\put(16.6,5.3){\circle*{0.07}}

\put(16,2.3){\circle*{0.07}} \put(16.3,2.3){\circle*{0.07}}
\put(16.6,2.3){\circle*{0.07}}
\end{picture}

\caption{Some of the resulting (row) vectors in \textsc{Sparse} instance computed
from the graph in Figure~\ref{fig:multi} by our reduction}
\label{fig:multi_matrix}
\end{figure}

We repeat this procedure by constructing $\ell$ vectors in
the upper level at each step until we finally construct the
set $V(1, \ell, d)$. Specifically, for $1 \leq k \leq d$ at
step $k$, we expand the seed strings that are used in the
previous step to get $\ell$ new seed strings of length $\ell^k$.
Concatenating $\ell^{d-k}$ of these and normalizing the
coordinates, we get the set of vectors $V(d-k+1, \ell, d)$.
It is clear that the pairwise dot products in such a set
are all $0$. Looking at two vectors $u \in V(i, \ell, d)$ and $v
\in V(j, \ell, d)$ for $i \neq j$, we also have that
their dot product is exactly $1/\ell$ since by construction
they have $\ell^{d-2}$ common nonzero coordinates and each
nonzero coordinate is $\ell^{(1-d)/2}$. Hence, we have $\ell d$
vectors in total forming the set $V(\ell, d)$ and all the
properties of an incoherent vector system are satisfied.
\end{proof}

We are ready to describe our reduction. Given the Smooth
Label Cover instance $L$, the derived multilayered instance
$L^{\ell}$ and an incoherent vector system $V(\ell,d)$ as
described above, we first construct $d=2^u7^{Tu}$ column
vectors for each $w \in X_{2j}$ for $j=1,\hdots,\ell/2$.
Specifically, given $w \in X_{2j}$, we define a vector for
each one of the $d$ labels, namely the set:

$$\{V_{2j}(1,\ell,d), V_{2j}(2,\ell,d), \hdots, V_{2j}(d,\ell,d)\}.$$

\noindent Note that the dot products of any two of these vectors is
$1/\ell$. Similar to the reduction presented in the previous
section, the actual column vectors of the matrix $\Phi$ are
composed of $|E^{\ell}|$ blocks, one for each hyper-edge
(Note that $|E^{\ell}| = |E| = O(n^{Tu})$). The column vector
$\Phi_{w,i}$ for $1 \leq i \leq d$ is defined as follows:
The blocks of the vector which correspond to the hyper-edges
that contain $w$ is $V_{2j}(i,\ell,d)$ and all the other
blocks consist of zeros. This reduction is in fact very similar
to the one described for the two-layered Label Cover problem.
Figure~\ref{fig:multi} is an example of
a simple instance of a multilayered Label Cover problem, and
Figure~\ref{fig:multi_matrix} illustrates the overall structure of
the matrix produced by the reduction. For simplicity, we
only show the vectors that belong to the vertices of the
first four layers.

The column vectors of the odd layers of the instance
$L^{\ell}$ are also defined similar to the previous section.
Given $v \in X_{2j-1}$ where $1 \leq j \leq \ell/2$, $\Phi_{v,i}$
consists of $|E^{\ell}|$ blocks for $i=1, \hdots 7^{(T+1)u}$.
The block of $\Phi_{v,i}$ corresponding to hyper-edge
$e \in E^{\ell}$ is $V_{2j-1}(\Pi_{e'}(i), \ell, d)$ where
$e' \in E$ is the edge that uniquely determines $e$ and $\Pi$
is the projection function of $L$. We define $y$ to be
the vector of all $1$s and $k=\sum_{i=1}^{\ell} |X_i| =
\frac{\ell}{2}(|V|+|W|)$. We
let $u$ and $T$ be constants and $\ell \gg d = 2^u 7^{Tu}$
be a constant. As usual, we extend $\Phi$ with an identity
matrix of appropriate size so as to satisfy the condition
of the sparse approximation problem. Note that the size of the
reduction is polynomial. In particular,
$M = \ell^d |E| = \ell^{2^u 7^{Tu}}5^u|W| = O(n^{Tu})$ and
$N = \frac{\ell}{2}(7^{(T+1)u}|V|+2^u7^{Tu}|W|) = O(n^{Tu})$ since
$u,T,\ell$ are all constants, and $|V|$ and $|W|$ are both $O(n^{Tu})$.

\emph{Proof of Theorem \ref{thm_a}:}
\begin{itemize}
\item \textit{Completeness:} Suppose that there is an assignment
which satisfies all the edges of $L$. Consider the same
assignment on $L^{\ell}$ on all layers, repeated $\ell/2$ times.
As noted before, this assignment also strongly satisfies all the
hyper-edges of $L^{\ell}$. Selecting all the $k$ column vectors
defined by this assignment, we see that the coordinates of a block
reserved for a hyper-edge  $e=(x_1,x_2,\hdots,x_{\ell})$ can be
covered by the column vectors corresponding to the labels assigned
to $x_1,x_2,\hdots,x_{\ell}$. Because, they form a set

$$
\{V_1(\Pi_{e'}(i),\ell,d), V_2(j,\ell,d), \hdots, V_{\ell-1}(\Pi_{e'}(i),
\ell,d), V_{\ell}(j,\ell,d)\},
$$

\noindent where $e' \in E$ is the edge
defining $e$, $\Pi_{e'}$ is the constraint on $e'$ and $i,j$ are
the labels assigned to the vertices on alternate layers. But,
$\Pi_{e'}(i)=j$ since $e$ is strongly satisfied. Hence, by the
definition of the incoherent vector system, we have that the
aforementioned set exactly covers the coordinates corresponding
to $e$. Since this is true for all the hyper-edges, the selected $k$ columns
can cover all the $M$ coordinates. In other words there is a
vector $x \in \mathbb{R}^{N}$ with $k$ nonzero entries such
that $\Phi x = y$.

\item \textit{Soundness:} Suppose that no assignment satisfies
more than a fraction $2^{-\gamma u}$ of the edges of $L$. Then,
as noted before, no assignment \emph{strongly} satisfies more
than a fraction $2^{-\gamma u}$ of the hyper-edges in $L^{\ell}$.
Our argument is similar to that of the two layered case. First,
in order to cover all the $M$ coordinates, one needs to select
exactly $1$ column vector from each vertex in all the layers.
Because, if there is a vertex $v$ from which no column vector
is selected, one cannot cover the blocks of hyper-edges incident
to $v$ by selecting one vector from each neighbor of $v$ (The
coordinates reserved for the layer that $v$ belongs to will never
be completely covered). Hence, one needs to select multiple
vectors to cover these blocks and it follows that all the
coordinates cannot be covered by selecting only $k$ column vectors.
It follows that it is sufficient to analyze the case where one
selects exactly $1$ column vector from each vertex in the graph.
In this case however, if there is a hyper-edge which is not strongly
satisfied, by definition of the incoherent vector system, the
block corresponding to that edge cannot be covered. Hence, for
all choices of $x \in \mathbb{R}^N$, we have that
$\|y - \Phi x\|_2 > 0$.
\end{itemize}

It remains to see what the coherence of $\Phi$ is. Similar to the
two layered case, the dot product of two column vectors from
distinct vertices in a layer is $0$, and the dot product of
two column vectors belonging to different layers is smaller than
$1/\ell$. Consider now the column vectors defined for the
vertices in layers of even index. They are analogous to $W$ in
the two layered case; there are $2^u 7^{Tu}$ of them for one vertex,
and any two column vectors corresponding to distinct labels
have dot product $1/\ell$ by construction. The same construction
goes through for the layers of odd index. However, there are
$7^{(T+1)u}$ column vectors for one vertex and hence there are
duplicates. Each constraint of a hyper-edge consists of $\ell-1$
constraints of the two layered instance. Recalling that hyper-edges
in $E^{\ell}$ are in one to one correspondence the edges in $E$
and that the smoothness property is satisfied for $L$, we conclude
that the amount of dot product coming from any two column
vectors corresponding to $a_1,a_2 \in \Sigma_V$ such that
$a_1 \neq a_2$ is at most $1/T$. Thus, the coherence of the
dictionary is upper bounded by $1/T+1/\ell$. Recalling that $\ell \gg d = 2^u 7^{Tu}$,
the dominating term in this expression is $1/T$. By choosing
$T$ arbitrarily large, it follows that \textsc{Sparse} is \textsf{NP}-hard
under dictionaries with coherence $\epsilon$ for arbitrarily
small $\epsilon > 0$. \hfill $\blacksquare$

\section{Reduction from the Multilayered Unique Label Cover}
The main hindrance that prevents us from pushing the coherence
below a constant is the term $1/T$. No matter how large
a function we assign to $\ell$, the coherence
remains at $1/T$ due to the construction. Ideally, one would like
to avoid possible duplications of vectors in the odd
layers of the multilayered instance, thereby completely getting
rid of $1/T$ and selecting $\ell$ as large as possible.
This is precisely what the constraint satisfaction problem
mentioned in the famous Unique Games Conjecture provides.
A Unique Label Cover instance $L$ is defined as follows:
$L = (G(V,W,E),R,\Pi)$ where

\begin{itemize}
\item $G(V,W,E)$ is a regular bipartite graph with vertex sets $V$
and $W$, and the edge set $E$. \vspace{1mm}

\item $\Sigma=\{1,2,\hdots,R\}$ is the label sets associated with
both $V$ and $W$. \vspace{1mm}

\item $\Pi$ is the collection of constraints on the edge set,
where the constraint on an edge $e$ is defined as a bijective
function $\Pi_e\colon \Sigma \rightarrow \Sigma$.
\end{itemize}

As usual, in the problem associated with a given instance,
the goal is to satisfy as many constraints as possible by
finding an assignment $A\colon V \rightarrow \Sigma$,
$A\colon W \rightarrow \Sigma$. The following is the well-known
conjecture by Khot \cite{Khot-UGC}:

\begin{con} (Unique Games Conjecture \cite{Khot-UGC})
Let $L$ be defined as above. Given $\epsilon,\delta>0$,
there exists a constant $R(\epsilon,\delta)$ such that it is \textsf{NP}-hard
to distinguish between the following two cases:

\begin{itemize}
\item YES. There is an assignment $A\colon V \rightarrow \Sigma$,
$A\colon W \rightarrow \Sigma$ such that at least $(1-\epsilon)$
fraction of the constraints in $\Pi$ are satisfied.

\item NO. No assignment can satisfy more than a fraction
$\delta$ of the constraints in $\Pi$.
\end{itemize}
\end{con}

It is possible that the conjecture is false, but the problem
Unique Games is still not in \textsf{P}. In fact, this is the version
stated in Theorem~\ref{thm_b}. We rule out the existence of
a polynomial time algorithm for \textsc{Sparse} under certain
conditions given that the problem described above is not in \textsf{P}.

Given a Unique Label Cover instance $L$, one can define an
$\ell$-layered Unique Label Cover instance as follows (This
is in same spirit to the one defined in previous section, see
Figure~\ref{fig:multilayered}): $L^{\ell} = (G^{\ell}(X_1,
\hdots, X_{\ell},E^{\ell}), R, \Pi^{\ell})$ where $\ell$ is
an even positive integer and

\begin{itemize}
\item $G^{\ell}(X_1, \hdots, X_{\ell},E^{\ell})$ is a multipartite
hyper-graph with vertex sets $X_1, X_2, \hdots, X_{\ell}$ and the
hyper-edge set $E$. \vspace{1mm}

\item $X_i = V$ for odd $i$, and $X_i = W$ for even $i$, where
$1 \leq i \leq \ell$. \vspace{1mm}

\item $E^{\ell}$ consists of all hyper-edges of the form
$(v,w,v,w,\hdots,v,w)$ on $\ell$ vertices belonging to $X_1, \hdots, X_{\ell}$,
respectively where $v \in V$, $w \in W$ and $(v,w) \in E$.

\item $\Pi^{\ell}$ is the collection of constraints on the
hyper-edge set, where the constraint on a hyper-edge
$e=(x_1,x_2,\hdots,x_{\ell-1},x_{\ell})=(v,w,\hdots,v,w)$ is itself
defined as a collection of $\ell-1$ constraints of the form
$\Pi_{(x_{2i-1},x_{2i})}\colon \Sigma \rightarrow \Sigma$ for $i=1,\hdots,\ell/2$, and
$\Pi_{(x_{2i+1},x_{2i})}\colon \Sigma \rightarrow \Sigma$ for $i=1,\hdots,\ell/2-1$, where
$x_i \in X_i$ for $i=1,\hdots,\ell$, and all the constraints are
identical to $\Pi_{(v,w)}$.
\end{itemize}

As usual, one is asked to find an assignment for all the
vertices of the instance, i.e. construct functions
$A_i\colon X_i \rightarrow \Sigma$ for $i=1, \hdots, \ell$,
where the goal is to satisfy as many hyper-edges as possible.
A constraint $\Pi_e$ is said to be \emph{strongly satisfied}
if all of the $\ell-1$ constraints defined for $e$ are satisfied.
The reduction from $L$ to $L^{\ell}$ is also
quite straightforward and is in the same spirit as the one
described in the previous section. Again, it is important here
to see that each hyper-edge in $E^{\ell}$ is uniquely
determined by an edge in $E$. Given the reduction, if there
is an assignment for $L$ which satisfies at least $(1-\epsilon)$
fraction of the constraints, then the same assignment
satisfies the set of hyper-edges that are uniquely determined
by the satisfied edges in $L$, i.e. at least $(1-\epsilon)$ of
the constraints in $L^{\ell}$ are \emph{strongly} satisfied. It
is also easy to see that if no assignment satisfies more than
a fraction $\delta$ of the constraints  in $L$, then no assignment
\emph{strongly} satisfies more than a fraction $\delta$
of the constraints in $L^{\ell}$, since at most
$(1-\delta)$ fraction of the constraints can be strongly
satisfied by the same reasoning we used for the completeness.

We now describe our reduction. Given the Unique
Label Cover instance $L$, the derived multilayered instance
$L^{\ell}$ and an incoherent vector system $V(\ell,R)$,
we first construct $R$ column
vectors for each $w \in X_{2j}$ for $j=1,\hdots,\ell/2$.
Specifically, given $w \in X_{2j}$, we define a vector for
each one of the $R$ labels, namely the set:

$$\{V_{2j}(1,\ell,R), V_{2j}(2,\ell,R), \hdots, V_{2j}(R,\ell,R)\}.$$

\noindent Note that the dot products of any two of these vectors is
$1/\ell$. The actual column vectors of the matrix $\Phi$ are
composed of $|E^{\ell}|$ blocks, one for each hyper-edge. The
column vector $\Phi_{w,i}$ for $1 \leq i \leq R$ is defined as follows:
The blocks of the vector which correspond to the hyper-edges
that contain $w$ is $V_{2j}(i,\ell,R)$ and all the other
blocks consist of zeros.

The column vectors of the odd layers of the instance
$L^{\ell}$ are defined as follows: Given $v \in X_{2j-1}$ where
$1 \leq j \leq \ell/2$, $\Phi_{v,i}$ consists of $|E^{\ell}|$ blocks
for $i=1, \hdots R$. The block of $\Phi_{v,i}$
corresponding to hyper-edge $e \in E^{\ell}$ is $V_{2j-1}(\Pi_{e'}(i),
\ell, R)$ where $e' \in E$ is the edge that uniquely defines
$e$ and $\Pi$ is the projection function of $L$. We
define $y$ to the vector of all $1$s, $k=\sum_{i=1}^{\ell} |X_i|$, and
extend $\Phi$ with an identity matrix as usual. Finally, we let
$\ell = n^C$ where $C(\epsilon)$ is a large constant depending on
$\epsilon>0$ and satisfying
$$
\max \left\{ |X_1|, \hdots, |X_{\ell}|, |E^{\ell}|\right\} \leq n^{C \epsilon}.
$$

\noindent Note that there exists such a constant since
$|X_i|$ and $|E^{\ell}|$ are of polynomial size. The size of the
reduction is also polynomial as
there are $N = \frac{\ell}{2}R(|V|+|W|)=\frac{R}{2}n^C(|V|+|W|)$ vectors with
$M = \ell^R |E|=n^{CR}|E|$
coordinates, where $C$ and $R$ are constants. With this choice of parameters,
we attain our goal of pushing the coherence below a constant.
Since there are no duplicate vectors for a given vertex in a
layer, the coherence stays at $1/\ell$. But, we have that
$$
k \leq {\max \left\{|X_1|, \hdots, |X_{\ell}|, |E^{\ell}|\right\}} \cdot n^C
\leq n^{C(1+\epsilon)} = \ell^{1+\epsilon}.
$$

\noindent Hence, $k^{\frac{1}{1+\epsilon}} \leq \ell < k$, which is to say that
$1/k < 1/\ell \leq k^{-\frac{1}{1+\epsilon}} = k^{-1+\epsilon'}$ for a suitably
chosen $\epsilon'$. There remains to check the completeness and the soundness
of the reduction.

\emph{Proof of Theorem \ref{thm_b}:}
\begin{itemize}
\item \textit{Completeness:} Suppose that there is an assignment
which satisfies at least $(1-\epsilon)$ fraction of the edges of $L$.
Consider the same assignment on $L^{\ell}$ on all layers, repeated
$\ell/2$ times. Since each hyper-edge in $E^{\ell}$ is uniquely
determined by an edge in $E$, this assignment strongly satisfies
at least $(1-\epsilon)$ fraction of the hyper-edges of $L^{\ell}$.
Selecting all the $k$ column vectors
defined by this assignment, we see that the coordinates of a block
reserved for a hyper-edge  $e=(x_1,x_2,\hdots,x_{\ell})$ can be
covered by the column vectors corresponding to the labels assigned
to $x_1,x_2,\hdots,x_{\ell}$. Because, they form a set

$$
\{V_1(\Pi_{e'}(i),\ell,d), V_2(j,\ell,d), \hdots, V_{\ell-1}(\Pi_{e'}(i),
\ell,d), V_{\ell}(j,\ell,d)\},
$$

\noindent where $e' \in E$ is the edge
defining $e$, $\Pi_{e'}$ is the constraint on $e'$ and $i,j$ are
the labels assigned to the vertices on alternate layers. But,
$\Pi_{e'}(i)=j$ since $e$ is strongly satisfied. Thus, at least
$(1-\epsilon)$ fraction of the $M$ coordinates of $y$ can be covered
upon a suitable choice of $x \in \mathbb{R}^N$ with all positive entries.
In other words, with such a choice, we have that $\Phi x$ has $1$s on
at least $(1-\epsilon)$ fraction of the coordinates and other coordinates
are between $0$ and $1$. This implies $\|y - \Phi x\|_2 \leq \sqrt{\epsilon M}$.

\item \textit{Soundness:} This requires a relatively more elaborate
analysis compared to the previous cases. Suppose that no assignment
satisfies more than a fraction $\delta$ of the edges of $L$. Take
a random hyper-edge $e=(v,w,\hdots,v,w)$ in $L^{\ell}$. Take all the
$\ell-1$ constraints forming the constraint function $\Pi_e$. Then,
the expected number of satisfied constraints among these is at most
$\delta(\ell-1)$. For simplicity of the argument, let us first
consider the case where all the constraints are not satisfied. In
this case, by selecting the same assignment for the vertices in odd
layers, we can cover half of the block corresponding to $e$.
Similarly, by selecting the same assignment for the vertices in
even layers, we can cover another half. Of course, there are overlaps
between these two halves. We call this assignment the
\emph{canonical assignment}, which selects exactly one column vector
from each vertex. Note that no other assignment selecting one column
vector from each vertex can do better than covering half of the
coordinates of $y$ by the even and odd layers separately. If
different assignments are used in different layers, there will be overlaps
in the coordinates of certain hyper-edges by our construction and
less than half of the coordinates will be covered in this case. We will argue
that it suffices to analyze the canonical assignment. In other words,
selecting multiple column vectors from vertices will not increase
the number of coordinates that can be covered.

Consider an assignment where we select multiple column vectors
from vertices. It is clear that such an assignment can be attained
starting from the canonical assignment and performing an interchange of
vectors by excluding a vector from a vertex $v$ and including a new vector
to another vertex $w$ thereby increasing the number of vectors we
select from $w$ by $1$. Assume
without loss of generality that $v$ is in an odd layer and
$w$ is in an even layer. As discussed in the previous
paragraph, the exclusion of the single vector from $v$ might result in a ``loss''
of at least $1/\ell-1/(2\ell) = 1/(2\ell)$ fraction of the coordinates of
the blocks incident to $v$. This is by the fact that the vector can cover
$1/\ell$ fraction of the coordinates and there might be at most $1/(2\ell)$
overlaps with vectors in even layers. Take a hyper-edge incident to $w$.
Let us analyze the ``gain'' contributed by the additional vector. We
can assume without loss of generality that the vertices of the hyper-edge
have at least one vector. Otherwise, we lose a fraction of $1/(2\ell)$
of its coordinates by the absence of a vector at each vertex making
the loss already greater than or equal to the gain. Then,
the fraction of overlaps between the new vector and the vectors in
the \emph{even} layers is $(1/\ell^2) \cdot (\ell/2) = 1/(2\ell)$. Thus,
the contribution of the new vector in $w$ is at most $1/\ell - 1/(2\ell)
= 1/(2\ell)$. Since the gain is not more than the loss, it
suffices to analyze the canonical assignment.

Let us now analyze the canonical assignment.
As noted, even with this assignment, there are overlaps
between indices covered by the odd layers and even layers since we
assumed that no edge is satisfied. We calculate the fraction of
overlaps as follows: Take a vertex of even layer, a hyper-edge $e$ incident
to this vertex and consider the block corresponding to $e$. The fraction of
overlaps of the column vector corresponding to this vertex with all
the ones on the odd layers is $(1/\ell^2) \cdot (\ell/2) =
1/(2\ell)$ by the properties of the incoherent vector system we use.
More explicitly, there are $\ell/2$ vertices of odd layers in a
hyper-edge and $1/\ell^2$ fraction of the coordinates overlap
with a vertex in one odd layer, assuming that all the $\ell-1$
constraints of the hyper-edge are not satisfied. Now, since there are
$\ell/2$ vertices of even layers and vertices in distinct even layers
cover distinct coordinates, it follows that the total
fraction of overlaps is $1/(2\ell) \cdot (\ell/2) = 1/4$. Thus,
the fraction of coordinates that can be covered for a hyper-edge
is $1/2+1/2- 1/4 = 3/4$. Since the fraction of satisfied constraints
of the hyper-edge is $\delta(\ell-1)$ in expectation, this can
only make an extra contribution proportional to $\delta$ to the number
of coordinates that can be covered. As a result, the coordinates
that cannot be covered for a hyper-edge is at least
$1/4-\Theta(\delta)$. Since our argument runs for a random hyper-edge,
taking all the $M$ coordinates formed by all the hyper-edges,
we have that $\|y - \Phi x\|_2 \geq \sqrt{(1/4-\Theta(\delta)) M}$.
\end{itemize}

The theorem now follows since the ratio of the values in the soundness and
the completeness can be arbitrarily large by arbitrarily small values of
$\epsilon$ and $\delta$. \hfill $\blacksquare$

\section{Final Remarks}
This article suggests that \textsc{Sparse} is only tractable for
very special classes of dictionaries. There remain the following open
problems we want to pose:
\begin{itemize}
\item Can the result based on Unique Games be strengthened to
\textsf{NP}-hardness? For $\mu = k^{-1+\epsilon}$, we can only rule out
a constant factor approximation due to the imperfect completeness.

\item Can one extend the complexity results in this article by
parameterizing with the Restricted Isometry Constant?

\item The results about OMP thus far only provide approximations.
Is it possible to find an exact solution for the problem by using
only the assumption $\mu = O(k^{-1})$? On the other hand,
a super constant approximation in the case $\mu = k^{-1+\epsilon}$
might also be possible. Our results do not rule out this.

\item UG-hardness of the problem suggests that there may be algorithmic
solution using semi-definite programming. Are there semi-definite
programming based algorithms for sparse approximation?
\end{itemize}

\textbf{Acknowledgment:} We thank Anna C. Gilbert for pointing out
some algorithmic results on sparse approximation and providing a
chart on the complexity of the problem. We also thank the anonymous
referees whose comments helped improve the presentation.


{
\bibliographystyle{plain}
\bibliography{reference}

\begin{thebibliography}{10}

\bibitem{LC}
S.~Arora, L.~Babai, J.~Stern, and Z.~Sweedyk.
\newblock The hardness of approximate optima in lattices, codes, and systems of
  linear equations.
\newblock {\em J. Comput. Syst. Sci.}, 54(2, part 2):317--331, 1997.

\bibitem{Pcp}
S.~Arora, C.~Lund, R.~Motwani, M.~Sudan, and M.~Szegedy.
\newblock Proof verification and the hardness of approximation problems.
\newblock {\em J. ACM}, 45(3):501--555, 1998.

\bibitem{Pcp2}
S.~Arora and S.~Safra.
\newblock Probabilistic checking of proofs: a new characterization of {NP}.
\newblock {\em J. ACM}, 45(1):70--122, 1998.

\bibitem{Candes}
E.~J. Candes, C.~E. Eldar, D.~Needell, and P.~Randall.
\newblock Compressed sensing with coherent and redundant dictionaries.
\newblock {\em Appl. Comput. Harmon. A.}, 31:59--73, 2011.

\bibitem{CSSP-UG}
A.~\c{C}ivril.
\newblock Column subset selection problem is {UG}-hard.
\newblock {\em J. Comput. Syst. Sci.}, 80(4):849--859, 2014.

\bibitem{Volume}
A.~\c{C}ivril and M.~Magdon-Ismail.
\newblock Exponential inapproximability of selecting a maximum volume
  sub-matrix.
\newblock {\em Algorithmica}, 65(1):159--176, 2013.

\bibitem{Davenport}
M.~A. Davenport and M.~B. Wakin.
\newblock Analysis of orthogonal matching pursuit using the restricted isometry
  property.
\newblock {\em IEEE T. Inform. Theory}, 56(9):4395--4401, 2010.

\bibitem{Davis}
G.~Davis, S.~Mallat, and M.~Avellaneda.
\newblock Adaptive greedy approximations.
\newblock {\em Constr. Approx.}, 13:57--98, 1997.

\bibitem{Lebesgue}
D.~L. Donoho, M.~Elad, and V.~N. Temlyakov.
\newblock On lebesgue-type inequalities for greedy approximation.
\newblock {\em J. Approx. Theory}, 147:185--795, 2007.

\bibitem{Feige}
U.~Feige.
\newblock A threshold of $\ln{n}$ for approximating set cover.
\newblock {\em J. ACM}, 45(4):634--652, 1998.

\bibitem{Gilbert}
A.~C. Gilbert, S.~Muthukrishnan, and M.~J. Strauss.
\newblock Approximation of functions over redundant dictionaries using
  coherence.
\newblock In {\em Proceedings of the 14th Annual ACM-SIAM Symposium on Discrete
  Algorithms ({SODA})}, pages 243--252, 2003.

\bibitem{SmoothLC}
S.~Khot.
\newblock Hardness results for coloring 3-colorable 3-uniform hypergraphs.
\newblock In {\em Proceedings of the 43rd Annual IEEE Symposium on Foundations
  of Computer Science ({FOCS})}, pages 23--32, 2002.

\bibitem{Khot-UGC}
S.~Khot.
\newblock On the power of unique 2-prover 1-round games.
\newblock In {\em Proceedings of the 34th Annual ACM Symposium on Theory of
  Computing ({STOC})}, pages 767--775, 2002.

\bibitem{Khot-Maxcut}
S.~Khot, G.~Kindler, E.~Mossel, and R.~O'Donnell.
\newblock Optimal inapproximability results for max-cut and other 2-variable
  {CSP}s?
\newblock {\em SIAM J. Comput.}, 37:319--357, 2007.

\bibitem{Liu-T}
E.~Liu and V.~N. Temlyakov.
\newblock The orthogonal super greedy algorithm and applications in compressed
  sensing.
\newblock {\em IEEE T. Inform. Theory}, 58(4):2040--2047, 2012.

\bibitem{Livshitz}
E.~D. Livshitz.
\newblock On the optimality of the orthogonal greedy algorithm for
  $\mu$-coherent dictionaries.
\newblock {\em J. Approx. Theory}, 164(5):668--681, 2012.

\bibitem{Lund-Yannakakis}
C.~Lund and M.~Yannakakis.
\newblock On the hardness of approximating minimization problems.
\newblock {\em J. ACM}, 41(5):960--981, 1994.

\bibitem{Natarajan}
B.~K. Natarajan.
\newblock Sparse approximate solutions to linear systems.
\newblock {\em SIAM J. Comput.}, 24(2):227--234, 1995.

\bibitem{Raz}
R.~Raz.
\newblock A parallel repetition theorem.
\newblock {\em SIAM J. Comput.}, 27(3):763--803, 1998.

\bibitem{Temlyakov1}
V.~N. Temlyakov.
\newblock Greedy algorithms and m-term approximation with regard to redundant
  dictionaries.
\newblock {\em J. Approx. Theory}, 98:117--145, 1999.

\bibitem{Temlyakov2}
V.~N. Temlyakov.
\newblock Weak greedy algorithms.
\newblock {\em Adv. Comput. Math.}, pages 213--227, 2000.

\bibitem{Zheltov}
V.~N. Temlyakov and P.~Zheltov.
\newblock On performance of greedy algorithms.
\newblock {\em J. Approx. Theory}, 163(9):1134--1145, 2011.

\bibitem{Tropp}
J.~A. Tropp.
\newblock Greed is good: Algorithmic results for sparse approximation.
\newblock {\em IEEE T. Inform. Theory}, 50(10):2231--2242, 2004.

\end{thebibliography}
}
\end{document}